\documentclass{amsart}
\usepackage{amsmath,amssymb}
\usepackage{graphicx,subfigure}
\newtheorem{theorem}{Theorem}
\newtheorem{proposition}[theorem]{Proposition}
\newtheorem{corollary}[theorem]{Corollary}

\newtheorem{lemma}[theorem]{Lemma}

\theoremstyle{definition}

\theoremstyle{remark}
\newtheorem{remark}[theorem]{Remark}

\newcommand{\al}{\alpha}

\newcommand{\de}{\delta}
\newcommand{\ep}{\epsilon}

\newcommand{\la}{\lambda}

\newcommand{\vp}{\varphi}

\newcommand{\De}{\Delta}

\newcommand{\Si}{\Sigma}

\def\CC{\mathbb{C}}
\def\NN{\mathbb{N}}
\def\RR{\mathbb{R}}

\renewcommand\SS{\mathbb{S}}
\newcommand{\cV}{{\mathcal V}}

\newcommand{\pd}{\partial}
\newcommand\minus\backslash

\newcommand\lan\langle
\newcommand\ran\rangle

\DeclareMathOperator\Real{Re}

\renewcommand\leq\leqslant
\renewcommand\geq\geqslant
\newlength{\intwidth}

\DeclareMathOperator\Imag{Im}

\newcommand\tvp{\widetilde\vp}

\begin{document}

\title[Dislocations of arbitrary topology in Coulomb
eigenfunctions]{Dislocations of arbitrary topology\\ in Coulomb
  eigenfunctions}

\author{Alberto Enciso}
\address{Instituto de Ciencias Matem\'aticas, Consejo Superior de
  Investigaciones Cient\'\i ficas, 28049 Madrid, Spain}
\email{aenciso@icmat.es, david.hartley@icmat.es, dperalta@icmat.es}

\author{David Hartley}

\author{Daniel Peralta-Salas}

%
%
\begin{abstract}
For any finite link $L$ in $\RR^3$ we prove the existence of a complex-valued eigenfunction of the Coulomb Hamiltonian such
that its nodal set contains a union of connected components
diffeomorphic to $L$. This problem goes back to Berry, who
constructed such eigenfunctions in the case where $L$ is the trefoil
knot or the Hopf link and asked the question about the general result.
\end{abstract}

\maketitle

\section{Introduction}

Dislocations were introduced in quantum mechanics by Berry and Nye in
1974 by analogy with the classical wavefront dislocations studied in
optics and have found numerous applications in very diverse areas of
science including chemistry, water waves and the theory of liquid
crystals~\cite{Nye}. If one writes the quantum mechanical wavefunction
in terms of its amplitude and phase as $\psi=\rho\, e^{i\chi}$, we
recall that a {\em dislocation}\/ is a connected component of the zero
set $\{\rho=0\}$ such that the phase changes by a nonzero multiple of
$2\pi$ on a closed circuit around it.

Motivated by problems in the theory of dislocations, in \cite{Be01}
Berry constructs eigenfunctions of the Coulomb Hamiltonian in $\RR^3$ whose
nodal set contains a trefoil knot or a Hopf link as a union of
connected components. He then raises the question as to whether there
exist eigenfunctions of a quantum system whose nodal set has
components with higher order linking. The existence of knotted
structures, both from theoretical \cite{Annals,Acta} and applied
\cite{De10,Irvine,Irvine2,Science} viewpoints, has recently attracted considerable
attention, especially in optics and in fluid mechanics.

This question was answered in the affirmative by the authors in
\cite{JEMS} where the quantum system considered was the harmonic
oscillator. In this paper we return to the Coulomb potential, as the
original setting considered by Berry, and prove that any link is the
union of connected components of the nodal set for some
eigenfunction. As we will see later on, the singularity of the
potential and the fact that one can no longer employ bound states of
arbitrarily high energy introduce serious technical complications in
the problem that must be dealt with using new ideas.

Eigenfunctions of the Coulomb Hamiltonian are the functions $\psi$ in the
Sobolev space $H^1(\RR^3)$ that satisfy the equation
\begin{equation}\label{E.hyd}
	\left(\Delta+\frac{2}{|x|}+\lambda\right)\psi=0,
\end{equation}
in $\RR^3$. It is well known that the eigenvalues are given by
\[
\la_n:=-\frac1{n^2}\,,
\]
for $n\in\NN$ and that a basis of the eigenspace corresponding
to~$\la_n$ is given by
\begin{align*}
 \psi_{nlm}&:=f_{nl}(r)\,Y_{lm}(\theta,\phi)\,,\\
 f_{nl}&:= A_{nl}\,e^{-r/n}\,r^l\,L_{n-l-1}^{2l+1}\Big(\frac{2r}{n}\Big)\,,
\end{align*}
where $0\leq l\leq n-1$, and $-l\leq m\leq l$.
Here $A_{nl}$ is a normalization factor that we take to be
\[
A_{nl}=\frac{2^{l-1}(n-l-1)!}{(n+l)!},
\]
$L_k^{\al}$ are the associated Laguerre polynomials, and $Y_{lm}$ are
the spherical harmonics on $\SS^2$. The degeneracy of the eigenspace
of the $\la_n$ energy level is therefore seen to be $n^2$.

The main result of the paper is the following theorem, which shows
that, as conjectured by Berry, there are Coulomb eigenfunctions having dislocations of arbitrary topology:

\begin{theorem}\label{T.main}
  Let $L$ be any finite link in $\RR^3$. Then there is some positive constant
  $E=E(L)$ such that, for any Coulomb eigenvalue 
  with $\la_n>-E$, there exist a complex-valued
  eigenfunction~$\psi$ of energy~$\la_n$ and a diffeomorphism~$\Phi$
  of~$\RR^3$ such that $\Phi(L)$ is a union of connected components of
  the zero set $\psi^{-1}(0)$. 
\end{theorem}

It should be emphasized that the deformed link $L':=\Phi(L)$ is a bona
fide dislocation set, meaning that there is a nonzero change of phase
(actually, of~$2\pi$) along a circuit around any component of~$L'$. This is an
immediate consequence of the fact that $L'$ satisfies the
nondegeneracy condition
\[
\text{rank}(\nabla \Real \psi(x),\nabla\Imag \psi(x))=2
\]
for all $x\in L'$, which means that $L'$ arises as the transverse intersection of the
zero sets of the real and imaginary parts of the eigenfunction~$\psi$. 
It is also worth mentioning that, as Berry conjectured, the link $L'$
is \textit{structurally stable}, that is, there exists an $\ep>0$
such that for any $C^1$ complex-valued function $\vp$ with
$\|\vp-\psi\|_{C^1}<\ep$, there is a diffeomorphism $\Phi_1$ of
$\RR^3$ close to the identity such that $\Phi_1(L')$ is a union of
connected components of the zero set $\vp^{-1}(0)$.

Heuristically, the idea of the proof of the theorem is the
following. As the energy levels $\la_n$ tend to 0 for large $n$, it is
clear that the formal limit of Equation~\eqref{E.hyd} as $n\to\infty$ is
\begin{equation}\label{E.hyd0}
	\left(\De+\frac2{|x|}\right)\vp=0\,.
\end{equation}
For this equation, one can prove that if the link~$L$ is contained in
a ball $B_{R_0}$ of a certain fixed radius (and this can always be ensured upon
deforming the link with a suitable diffeomorphism), then there is a
solution of~\eqref{E.hyd0} whose zero set contains a union of
connected components that is a small deformation of the
link~$L$. Furthermore, this union is a structurally stable set. The
point now is that one can prove that there is a sequence of
eigenfunctions of energy~$\la_n$ that approximate the above
function~$\vp$ in~$B_{R_0}$ as $n\to\infty$, so by structural stability we
infer that for all large enough~$n$ there are eigenfunctions with
energy~$\la_n$ and whose zero set contains a union of connected
components diffeomorphic to~$L$.

This approach is completely different from Berry's, which hinges in
controlling explicit perturbations of a fixed eigenfunction of the
Coulomb Hamiltonian whose zero set contains a degenerate circle (i.e.,
a circle on which the gradient of the eigenfunction is zero
identically). Berry's construction suggests that one probably needs
eigenfunctions with high~$n$ to realize complicated knots; e.g., for
the trefoil knot he needs at least $n=7$. Concerning the relationship
between Theorem~\ref{T.main} and our previous result for the harmonic
oscillator~\cite{JEMS}, and leaving technicalities aside, an
intriguing fact is that whereas for the harmonic oscillator one shows
the existence of arbitrary links only if they are contained in
arbitrarily small balls (specifically, of radius of order $n^{-1/2}$), in the
case of the Coulomb potential any link can be realized with a diameter
of order~1. In a way, this can be understood as evidence of the fact
that the key ingredient of the proof is only the degeneracy of the
eigenfunctions, rather than having arbitrarily large energies.



The structure of the paper is as follows. In Section \ref{S.largen} we
consider the behavior of the eigenfunctions $\psi_{nlm}$ for large
values of $n$ by computing the high-order asymptotics of the Laguerre
polynomials, and prove that they become $C^1$-close to a continuous
solution of (\ref{E.hyd0}). Section \ref{S.approx} details the
existence of continuous solutions to (\ref{E.hyd0})
in~$\RR^3$ that approximate functions satisfying (\ref{E.hyd0}) in a
compact subset of the ball $B_{R_0}$, for a certain fixed
radius~$R_0$. Finally, in Section \ref{S.main}, a function satisfying
(\ref{E.hyd0}) in $\RR^3$ with a nodal set containing a small deformation of the link~$L\subset B_{R_0}$ is constructed and the two previous approximation results are used to prove Theorem~\ref{T.main}.

\section{Large~$n$ asymptotics for Coulomb eigenfunctions}
\label{S.largen}

In this section we will present a large~$n$ asymptotic expansion
for the radial part of the eigenfunctions of the Coulomb Hamiltonian, which will be useful in the
proof of the main theorem. The result can be stated as follows, where
$J_\nu$ denotes the Bessel function of the first kind:

\begin{proposition}\label{P.asympt}
For any fixed~$l$ and~$R$, 
\[
\lim_{n\to\infty} \bigg\| f_{nl} - \frac{J_{2l+1}(\sqrt{8r})}{\sqrt{8r}}\bigg\|_{C^0((0,R))}=0\,.
\]
\end{proposition}
\begin{proof}
The proposition follows from \cite[Theorem 8.22.4]{Szego75}, which
gives the asymptotic expansion for fixed~$\al>0$:
\begin{equation}\label{LagAsymp}
	 e^{-x/2}x^{\frac{\alpha}{2}}L_k^{\alpha}(x)=\frac{\Gamma(k+\alpha+1)}{(k+\frac{\alpha+1}{2})^{\frac{\alpha}{2}}k!}J_{\alpha}\left(\sqrt{(4k+2\alpha+2)x}\right)+x^{\frac{5}{4}}O\left(k^{\frac{2\alpha-3}{4}}\right),
\end{equation}
where the bound holds uniformly in $0\leq x\leq R$ for any $R>0$. Using the substitutions $x=2r/n$, $k=n-l-1$, and $\al=2l+1$ we are able to obtain the asymptotic expansion of $f_{nl}$:
\begin{align*}
f_{nl}(r)=&\frac{A_{nl}n^{l+\frac12}}{2^{l+\frac12}\sqrt{r}}e^{-r/n}\left(\frac{2r}{n}\right)^{\frac{2l+1}{2}}L_{n-l-1}^{2l+1}\left(\frac{2r}{n}\right)\\
=&\frac{A_{nl}n^{l+\frac12}}{2^{l-1}\sqrt{8r}}\left(n^{-l-\frac{1}{2}}\frac{(n+l)!}{(n-l-1)!}J_{2l+1}\left(\sqrt{8r}\right)+r^{\frac{5}{4}}n^{-\frac54}O\left(n^{\frac{4l-1}{4}}\right)\right)\\
=&\frac{J_{2l+1}(\sqrt{8r})}{\sqrt{8r}}+r^{\frac{3}{4}}O\left(n^{-2}\right),
\end{align*}
from which the proposition follows. To show that the error is of order
$n^{-2}$ we have used the definition of the factors $A_{nl}$ and
Stirling's asymptotic formula for the factorial, which yields:
\[
\frac{(n-l-1)!}{(n+l)!}=n^{-2l-1}+O(n^{-2l-2})\,.
\]
\end{proof}

Of course, since the eigenfunctions satisfy the radial equation
\[
\bigg(\pd_r^2+\frac2r\pd_r-\frac{l(l+1)}{r^2}+\frac2r+\la_{n}\bigg) f_{nl}(r)=0\,,
\]
this uniform estimate can be easily promoted to a $C^k$ bound. For our
purposes, it is enough to state the result as follows:

\begin{corollary}\label{C.asympt}
For any fixed~$l$ and any $0<R_1<R_2$, 
\[
\lim_{n\to\infty} \bigg\| f_{nl} - \frac{J_{2l+1}(\sqrt{8r})}{\sqrt{8r}}\bigg\|_{C^1((R_1,R_2))}=0\,.
\]
\end{corollary}
\begin{proof}
Since the difference
\[
g_{nl}:= f_{nl} - \frac{J_{2l+1}(\sqrt{8r})}{\sqrt{8r}}
\]
satisfies the ODE
\[
\bigg(\pd_r^2+\frac2r\pd_r-\frac{l(l+1)}{r^2}+\frac2r+\la_{n}\bigg) g_{nl}(r)=\la_{n}\frac{J_{2l+1}(\sqrt{8r})}{\sqrt{8r}}\,,
\]
where the coefficients and the source term are uniformly bounded for
fixed~$l$ and all~$n$ on any finite interval $(R_1,R_2)$ with $R_1>0$,
it is standard that 
\begin{multline*}
\bigg\| f_{nl} -
\frac{J_{2l+1}(\sqrt{8r})}{\sqrt{8r}}\bigg\|_{C^1((R_1,R_2))}\leq C
\bigg\| f_{nl} -
\frac{J_{2l+1}(\sqrt{8r})}{\sqrt{8r}}\bigg\|_{C^0((R_1,R_2))} \\
+ C|\la_{n}|\bigg\|\frac{J_{2l+1}(\sqrt{8r})}{\sqrt{8r}}\bigg\|_{C^0((R_1,R_2))}\,,
\end{multline*}
where the constant depends on $R_1$ and $R_2$ but not on~$n$. Since
$\la_{n}\to0$ as $n\to\infty$, the result then follows from the $C^0$
bound given in Proposition \ref{P.asympt}.
\end{proof}

\section{An approximation theorem for the Coulomb problem with zero
  energy}
\label{S.approx}

In this section we shall prove an approximation theorem for the zero-energy
Coulomb Hamiltonian $\De+2/|x|$. In comparison with other Runge
theorems, a significant technical difference is that the operator is
not negative, which will force us to establish an approximation
theorem for sets contained in balls of radii not larger than
\[
R_0:=\frac{\sqrt\pi}4\,.
\]

For this, a first step is to show the existence of a suitable Green's
function for the Coulomb Hamiltonian:

\begin{lemma}\label{L.GF}
Let us consider the open ball $B_R\subset\RR^3$ centered at the origin and
of radius $R<R_0$. There exists a symmetric Dirichlet Green's function
on $B_R\times B_R$, bounded as
\[
|G_R(x,y)|\leq \frac{C}{|x-y|}
\]
and continuous outside
of the diagonal, which satisfies
\[
\bigg(\De_x + \frac2{|x|}\bigg) G_R(x,y)=-\de(x-y)\,,\qquad
G_R(\cdot,y)|_{\pd B_R}=0
\]
on $B_R\times B_R$.
\end{lemma}
\begin{proof}
The result will follow from the fact that the Dirichlet spectrum of ${-\De-\frac2{|x|}}$ is positive on $B_R$. To prove this we can use the min-max principle since its eigenfunctions are in $H_0^1$. We first note that for all $\vp\in C_0^{\infty}(B_R)$ we have
\begin{align*}
\int_{B_R}\frac{\vp^2}{|x|}\leq& \left(\int_{B_R}\vp^2\right)^{\frac12}\left(\int_{B_R}\frac{\vp^2}{|x|^2}\right)^{\frac12}\\
\leq& \left(\int_{B_R}\vp^2\right)^{\frac12}\left(4\int_{B_R}|\nabla\vp|^2\right)^{\frac12}\\
\leq& \frac{2}{\sqrt{\la_1(B_R)}}\int_{B_R}|\nabla\vp|^2,
\end{align*}
where 
\[
\la_1(B_R):=\frac\pi {R^2}
\]
is the first Dirichlet eigenvalue of the Laplacian in the ball of
radius~$R$ and we have used Hardy's inequality followed by
Poincar\'e's inequality. Therefore the first Dirichlet eigenvalue of
$-\De-\frac2{|x|}$ on $B_R$ satisfies:
\begin{align*}
\la&=\min_{\vp\in H_0^1(B_R)\backslash\{0\}}\frac{\int_{B_R}\Big(|\nabla\vp|^2-\frac2{|x|}\vp^2\Big)}{\int_{B_R}\vp^2}\\
&\geq\left(1-\frac{4R}{\sqrt{\pi}}\right)\min_{\vp\in H_0^1(B_R)\backslash\{0\}}\frac{\int_{B_R}|\nabla\vp|^2}{\int_{B_R}\vp^2}\\
&>0.
\end{align*}
Reference~\cite[Theorems 2 and 8]{Zhao} then ensures the existence of
a Green's function $G_R$ as above.
\end{proof}

We are now ready to prove the approximation theorem for sets contained
in the ball $B_{R_0}$. Notice that the series~$\tvp$ defined in the statement below satisfies
the equation
\[
\bigg(\De+\frac2{|x|}\bigg)\, \tvp=0
\]
in the whole space $\RR^3$, is continuous at the origin and falls off
at infinity as $|x|^{-3/4}$.

\begin{theorem}\label{T.approx}
Let us consider a radius $R<R_0$ and fix an integer~$k$ and a positive
real $\ep$. Suppose that $K$ is a compact subset of the ball $B_R$
which does not contain the origin and whose complement $B_R\backslash
K$ is connected. If a function~$\vp$ satisfies the equation
\[
\bigg(\De+\frac2{|x|}\bigg)\, \vp=0
\]
in a neighborhood~$U$ of~$K$, then for any large enough~$N$ there is a finite series of the form
\[
\tvp:=\sum_{l=0}^N\sum_{m=-l}^l c_{lm}\,
\frac{J_{2l+1}(\sqrt{8r})}{\sqrt{8r}}\, Y_{lm}(\theta,\phi),
\]
where $c_{lm}$ are constants,
which approximates~$\vp$ in the sense that
\[
\|\vp-\tvp\|_{C^k(K)}<\ep\,.
\]
\end{theorem}
\begin{proof}
Let us take a smooth function $\chi:\RR^3\rightarrow\RR$ such that
$\chi=1$ in a neighborhood of $K$ and $\chi=0$ outside the
neighborhood $U$ of $K$, whose closure can be assumed to lie in
$B_R\backslash\{0\}$, and define a smooth extension of $\vp$ to
$\RR^3$ by setting $\vp_1:=\chi \vp$. Since $\vp_1|_{\partial B_R}=0$,
we can use the Green's function in Lemma \ref{L.GF} to represent
$\vp_1$ as the integral
\begin{equation}\label{vp1G}
\vp_1(x)=\int_{B_R}G_R(x,y)\rho(y)\,dy,
\end{equation}
where $\rho(x)=-\De \vp_1(x)-\frac2{|x|} \vp_1(x)$ is supported in $U\backslash K$. In order to see
this, notice that the difference
\[
v:=\vp_1(x)-\int_{B_R}G_R(x,y)\rho(y)\,dy
\]
is zero on $\pd B_R$ and satisfies the equation
\[
\bigg(\De+\frac2{|x|}\bigg)\, v=0\,.
\]
Since we saw in the proof of Lemma~\ref{L.GF} that the operator $\De+2/|x|$ in $B_R$ with Dirichlet
boundary conditions and $R<\sqrt\pi/4$ is negative definite, we infer
that $v=0$, thereby proving~\eqref{vp1G}.

Since the support of $\rho$ is contained in $U\backslash K$, a
continuity argument ensures we can approximate the integral uniformly
in $K$ by a finite Riemann sum, i.e. for any $\de>0$ there is a large
enough integer $J$, points $x_j\in U\backslash K$ and constants $a_j$ such that
\[
\vp_2(x)=\sum_{j=1}^Ja_j\, G_R(x,x_j),
\]
satisfies
\begin{equation}\label{E.vpapprox1}
\|\vp_2-\vp\|_{C^0(K)}=\|\vp_2-\vp_1\|_{C^0(K)}<\de.
\end{equation}

Next we choose $0<R_1<R$ such that the closure of $U$ is contained in $B_{R_1}$ and consider the infinite dimensional space
\[
\cV=\text{span}\{G_R(z,x):z\in B_R\backslash \overline{B_{R_1}}\}\subset C^0(K).
\]
By the Riesz--Markov theorem the dual space to $C^0(K)$, which we will
denote by $C^0(K)^*$, is the space of finite signed Borel measures on $B_R$ whose support is contained in the set $K$. We let $\mu\in C^0(K)^*$ be any measure such that $\cV$ is in its null space, i.e. $\int_{B_R}f(x)\,d\mu(x)=0$ for all $f\in \cV$, and define
\[
F(x):=\int_{B_R}G_R(x,y)\,d\mu(y).
\]
Notice that $F$ is well defined because the Green's function is
bounded as $|G_R(x,y)|\leq C/|x-y|$. In particular, $F$ is continuous in $B_R\backslash K$.

By the definition of $\mu$, $F(x)=0$ for all $x\in B_R\backslash \overline{B_{R_1}}$ and, by the properties of the Green's function,
\[
\De F+\frac2{|x|}F=-\mu.
\]
In particular, $\De F(x)+\frac2{|x|}F(x)=0$ distributionally for all
$x\in B_R\backslash K$. Hence, $F$ is analytic in the region
$B_R\backslash(K\cup\{0\})$, which is connected and contains an open
set, $B_R\backslash \overline{B_{R_1}}$, on which $F$ is identically
zero. Therefore $F(x)=0$ for all $x\in B_R\backslash K$ and it follows that $\vp_2$ is also in the null space of $\mu$:
\begin{align*}
	\int_{B_R}\vp_2(x)\,d\mu(x)&=\sum_{j=1}^Ja_j\int_{B_R}G_R(x_j,x)\,d\mu(x)\\
	&=\sum_{j=1}^Ja_j\,F(x_j)\\
	&=0.
\end{align*}
Notice that in the last line one gets zero because each $x_j$ lies in
$U\backslash K\subset B_R\backslash K$. We have also used that the
Green's function is symmetric.

Therefore $\vp_2$ cannot be separated from the subspace $\cV$. Hence by the Hahn--Banach theorem $\vp_2$ can be uniformly approximated in $K$ by functions in $\cV$, i.e. for any $\de>0$ there exists $\vp_3\in \cV$ such that
\begin{equation}\label{E.vpapprox2}
\|\vp_3-\vp_2\|_{C^0(K)}<\de.
\end{equation}
Now, by construction, $\vp_3$ satisfies
\[
\De \vp_3+\frac2{|x|}\vp_3=0,
\]
inside $B_{R_1}$. By the properties of $G_R(x,y)$ in Lemma \ref{L.GF}
$\vp_3$ is continuous at the origin, so it can be expanded as
\[
\vp_3(x)=\sum_{l=0}^{\infty}\sum_{m=-l}^lc_{lm}\frac{J_{2l+1}(\sqrt{8r})}{\sqrt{8r}}\, Y_{lm}(\theta,\phi),
\]
where $c_{lm}$ are constants. 

Since the series converges in $L^2(B_{R_1})$, for any $\de>0$ there is an integer $N$ such that the finite sum
\[
\tvp(x)=\sum_{l=0}^{N}\sum_{m=-l}^l c_{lm}\frac{J_{2l+1}(\sqrt{8r})}{\sqrt{8r}}\, Y_{lm}(\theta,\phi),
\]
satisfies
\[
\|\tvp-\vp_3\|_{L^2(B_{R_1})}<\de.
\]
As
\[
\De(\tvp-\vp_3)+\frac2{|x|}(\tvp-\vp_3)=0,
\]
in $B_{R_1}$, the fact that the closure of $U$ is contained in
$B_{R_1}\backslash\{0\}$ allows us to use standard elliptic estimates to promote
the above $L^2$~bound to the uniform estimate
\[
\|\tvp-\vp_3\|_{C^0(U)}<C\de.
\]
From this inequality and the bounds in (\ref{E.vpapprox1}) and (\ref{E.vpapprox2}) we obtain
\[
\|\tvp-\vp\|_{C^0(K)}<(C+2)\de.
\]
Finally, since
\[
\De(\tvp-\vp)+\frac2{|x|}(\tvp-\vp)=0,
\]
in a neighborhood of $K$, standard elliptic estimates then allow us to promote the $C^0$ bound to a $C^k$ bound, i.e.
\[
\|\tvp-\vp\|_{C^k(K)}<{C'}\de.
\]
The theorem then follows by choosing $\de$ small enough so that $C'\de<\ep$.
\end{proof}

\section{Coulomb eigenfunctions with dislocations of arbitrary topology}
\label{S.main}


This section will be devoted to the proof of Theorem~\ref{T.main}. For
this, we start by taking a diffeomorphism of $\RR^3$, $\Phi_0$, such
that the transformed link $L_0:=\Phi_0(L)$ is contained in a punctured
ball $B_R\backslash\{0\}$ for some $R<R_0$. Whitney's approximation theorem
guarantees that, by perturbing it if necessary, we can assume that $L_0$ is a real analytic submanifold of $\RR^3$.

Let $L_{0,a}$, $a=1,\ldots,M$, denote the connected components of
$L_0$ (which are just closed curves in $\RR^3$) and construct surfaces
$\Si_a^1\subset B_R\backslash\{0\}$ and $\Si_a^2\subset
B_R\backslash\{0\}$ that intersect transversally at $L_{0,a}$. This
can be achieved by considering an analytic submersion
$\Theta_a:W_a\rightarrow\RR^2$, where $W_a$ is a small tubular
neighborhood of $L_{0,a}$ and $\Theta_a^{-1}(0)=L_{0,a}$. We recall
that the map $\Theta_a$ is a {\em submersion}\/ if the rank of the
differential $(D\Theta_a)_x$ is~2 for all $x\in W_a$. The
existence of such a submersion is guaranteed since any closed curve in
$\RR^3$ has trivial normal bundle \cite{Ma59}. The surfaces can then
be taken as $\Si_a^1:=\Theta_a^{-1}((-1,1)\times\{0\})$ and
$\Si_a^2:=\Theta_a^{-1}(\{0\}\times(-1,1))$, so they are small
cylinders embedded in $W_a$.

With $b=1,2$, we now consider the Cauchy problems
\[
\De \vp_a^b+\frac{2}{|x|}\vp_a^b=0,\ \ \ \vp_a^b|_{\Si_a^b}=0,\ \ \ \partial_{\nu}\vp_a^b|_{\Si_a^b}=1,
\]
where $\partial_{\nu}$ denotes a normal derivative at the relevant
surface. Because the equations are analytic inside
$B_R\backslash\{0\}$, we can use the Cauchy--Kowalewski theorem to
obtain solutions in the closure of small neighborhoods of each
surface, $U_a^b\supset\Si_a^b$. By shrinking the neighborhoods if
necessary we can assume that $U_a^b\subset B_R\backslash\{0\}$ and the
tubular neighborhoods $U_a^1\cap U_a^2$ of $L_{0,a}$ are disjoint. 

We now take the union of these tubular neighborhoods,
\[
U:=\bigcup_{a=1}^{M}U_a^1\cap U_a^2,
\]
and define a complex valued function $\vp$ on the set $U$ as
\[
\vp|_{U_a^1\cap U_a^2}:=\vp_a^1+i\vp_a^2.
\]
By construction $\vp$ has the following properties:
\begin{enumerate}
\item It satisfies the equation $\De \vp+\frac{2}{|x|}\vp=0$ in the tubular neighborhood $U$ of the link $L_0$,
\item By taking $U$ small enough its nodal set is precisely $L_0$, i.e. $\vp^{-1}(0)=L_0$, and 
\item The intersection of the zero sets of the real and imaginary
  parts of~$\vp$ on $L_0$ is transverse, i.e. $\text{rank}(\nabla
  \Real \vp(x),\nabla\Imag \vp(x))=2$ for all $x\in L_0$. This is an
  immediate consequence of the Cauchy data used to construct $\vp_a^b$
  (namely, $\vp_a^b|_{\Si_a^b}=0$ and $\pd_\nu\vp_a^b|_{\Si_a^b}=1$)
  and of the fact that $\Si_a^1$ and $\Si_a^2$ intersect transversally.
\end{enumerate}

We denote by $K\subset U$ a compact set containing $L_0$ whose
complement $B_R\backslash K$ is connected. We can now use Theorem \ref{T.approx} to obtain constants $c_{lm}\in\CC$ and $N\in\NN$ such that the function
\[
\tvp=\sum_{l=0}^{N}\sum_{m=-l}^l c_{lm}\,\frac{J_{2l+1}(\sqrt{8r})}{\sqrt{8r}}\, Y_{lm}(\theta,\phi),
\]
satisfies
\begin{equation}\label{E.approx}
\|\tvp-\vp\|_{C^1(K)}<\ep,
\end{equation}
for an $\ep>0$ to be chosen later. With~$n$ a large number to be fixed
later, we now define
\[
\psi_n:=\sum_{l=0}^{N}\sum_{m=-l}^l c_{lm}\,f_{nl}(r)\, Y_{lm}(\theta,\phi),
\]
and note that
\begin{align*}
\|\psi_n-\tvp\|_{C^1(K)}&\leq \sum_{l=0}^{N}\sum_{m=-l}^l |c_{lm}|\left\|\left(f_{nl}-\frac{J_{2l+1}(\sqrt{8r})}{\sqrt{8r}}\right)Y_{lm}\right\|_{C^1(K)}\\
&\leq C\sum_{l=0}^{N}\left\|f_{nl}-\frac{J_{2l+1}(\sqrt{8r})}{\sqrt{8r}}\right\|_{C^1(K)}.
\end{align*}
By Corollary \ref{C.asympt} there exists $N'\in\NN$ such that for all $n\geq N'$ we have
\[
\|\psi_n-\tvp\|_{C^1(K)}<\ep,
\]
and combining this with Equation (\ref{E.approx}) we obtain
\begin{equation}\label{E.asympt}
\|\psi_n-\vp\|_{C^1(K)}<2\ep,
\end{equation}
for all $n\geq N'$

Item (iii) above and Thom's isotopy theorem~\cite[Theorem 20.2]{AR}
ensure that the link $L_0$ is structurally stable. We can therefore choose $\ep>0$ small enough so that for each $n\geq N'$ there exists a diffeomorphism, $\Phi_1$, of $\RR^3$ that is $C^1$-close to the identity and different from the identity just in a small neighborhood of $L_0$, such that $\Phi_1(L_0)$ is a union of connected components of the zero set $\psi_n^{-1}(0)$. Moreover, we can assume $\Phi_1(L_0)\subset B_R\backslash\{0\}$ and it is structurally stable since the $C^1(K)$-closeness of $\psi_n$ to $\vp$ implies that $\text{rank}(\nabla \Real \psi_n(x),\nabla\Imag \psi_n(x))=2$ for all $x\in \Phi_1(L_0)$. The theorem follows by setting $E:=-\la_{N'}$ and $\Phi:=\Phi_1\circ\Phi_0$.

\begin{remark}
It follows from the proof of Theorem~\ref{T.main} that if the link $L$ is contained in a punctured
ball $B_R\backslash\{0\}$ for some $R<R_0$, the diffeomorphism $\Phi$ can be chosen as close to the identity as one wishes in the $C^1$ norm.
\end{remark}

\section*{Acknowledgments}

The authors are supported by the ERC Starting Grants~633152 (A.E.) and~335079
(D.H.\ and D.P.-S.). This work is supported in part by the
ICMAT--Severo Ochoa grant
SEV-2015-0554.

\bibliographystyle{amsplain}

\end{document}